\documentclass[submission,copyright,creativecommons]{eptcs}
\providecommand{\event}{GASCom 2026} % Name of the event you are submitting to

\usepackage{iftex}

\ifpdf
  \usepackage{underscore}         % Only needed if you use pdflatex.
  \usepackage[T1]{fontenc}        % Recommended with pdflatex
\else
  \usepackage{breakurl}           % Not needed if you use pdflatex only.
\fi

\usepackage{xcolor}
\usepackage[utf8]{inputenc}
\usepackage[linesnumbered,ruled,vlined]{algorithm2e}
\usepackage{amsmath,amssymb,amsthm}
\usepackage{wrapfig}
\usepackage{cleveref}
\usepackage{mathrsfs}

\newtheorem{lemma}{Lemma}
\newtheorem{theorem}{Theorem}
\newtheorem{remark}{Remark}

\ifpdf
  \usepackage{underscore}         % Only needed if you use pdflatex.
  \usepackage[T1]{fontenc}        % Recommended with pdflatex
\else
  \usepackage{breakurl}           % Not needed if you use pdflatex only.
\fi

\title{Entropic Generation of Binary Words}
\author{Olivier Bodini%\orcidlink{0000-0002-1867-667X}
\institute{EREN\\ Université Sorbonne Paris-Nord}
\email{olivier.bodini@univ-paris13.fr}
\and
Francis Durand%\orcidlink{0009-0004-4146-0289}
\institute{EREN\\ Université Sorbonne Paris-Nord}
\institute{LIP6\\
Sorbonne Université, France}
\email{francis.durand@ens-paris-saclay.fr}
}

\def\titlerunning{Entropic Generation of Binary Words}
\def\authorrunning{O. Bodini \& F. Durand}

\begin{document}

\maketitle

\begin{abstract}
The uniform generation of $k$ Hamming weight binary words, equivalent to sampling $k$-subsets from $n$ elements, relies on random bits, which can be expensive.
We introduce a novel paradigm--\emph{random bit recycling}--and use it to generate such binary words in linear time while consuming as few random bits as possible.
The resulting algorithm is nearly optimal in terms of random bit consumption, meaning that it closely matches the \emph{Shannon entropic lower bound} coming from information theory.
\end{abstract}

\section{Introduction}

The uniform generation of combinatorial structures is a cornerstone of computer science.
In particular, the generation of binary words of length $n$ with exactly $k$ ones (or equivalently, random subsets of size $k$ from a set of $n$ elements) is an important primitive.
While the time complexity of these generators has been studied extensively, a secondary cost metric has gained prominence in recent years: the consumption of random bits.

In many computing environments, pseudo-random bits are cheap. 
However, in high-security applications, cryptography, and embedded systems, randomness is often derived from hardware true random number generators (TRNGs) \cite{chai2025survey}.
In these contexts, random bits are an expensive resource and can be a bottleneck for random generation.
When generating a combinatorial structure uniformly at random from a finite set of $F$ possibilities, the information-theoretic entropy~\cite{shannon1948mathematical} of the generated object is $\log_2 F$. 
By Shannon's source coding theorem, this entropy establishes a lower bound: any algorithm relying on fair coin flips (unbiased random bits) must consume, on average, at least $\mathcal{H} = \log_2 F$ bits to produce a perfectly uniform sample.
For the problem at hand of generating binary words of length $n$ with exactly $k$ ones, the size of the state space is $\binom{n}{k}$. 
Consequently, we strive to give algorithms whose expected random bit consumption is asymptotically close to $\log_2 \binom{n}{k}$.

It is important to emphasize that all algorithmic complexities presented in this work--both in terms of execution time and random bit consumption--are stated as \textit{expected} (average-case) bounds.
This is a necessity as generating a uniform random integer over a range whose size is not a power of 2 (for instance, a uniform integer in $\{1, 2, 3\}$) from a stream of fair coin flips cannot be achieved in bounded time~\cite{Knuth1976TheCO}. 
We call an algorithm \emph{first order entropic} (resp. near-entropic) if its random bit consumption is asymptotically in $\mathcal{H} + o(\mathcal{H})$ (resp. in $(1+\varepsilon)\mathcal{H} + o(\mathcal{H})$, for any $\varepsilon > 0$) where $\mathcal{H} = \log_2 \binom{n}{k}$ is the entropy of our structure.

Generating fixed Hamming weight binary words is a two-parameter problem defined by length $n$ and Hamming weight $k$ (where we can assume $k \le n/2$ by symmetry). 
To analyze our algorithms asymptotically, we evaluate complexities as $n \to \infty$ along the trajectory $k \sim n^\beta$ for a fixed $\beta \in (0, 1)$. We call this setting the \textit{polynomially sparse regime}.
Current generation methods generally bifurcate based on the number of $1$s. For highly sparse inputs ($k \ll n$), the partial Fisher-Yates shuffle~\cite{FisherYates} is efficient.
Conversely, for denser inputs ($k = \Omega(n^{2/3 + \varepsilon})$ for any $\varepsilon > 0$), an adaptation of the \textsc{Merge} algorithm~\cite{bacher2015mergeshuffle} performs well by initially approximating the target word with Bernoulli($k/n$) variables and subsequently correcting for bias.

To the best of our knowledge, no prior algorithm simultaneously achieves linear time complexity $O(n)$ and near-optimal entropic cost in the polynomially sparse regime with $\beta \in (0,2/3]$.

\subsection*{Our Contribution}

In this paper, we present a novel algorithm that bridges this gap. 
We propose a generator for uniform binary words that is both entropically near-optimal and computationally linear.

We observe that the standard Fisher-Yates procedure, while inserting the $k$ ones among the $n-k$ zeros, implicitly constructs a uniform random permutation on these $1$s.
In order to capitalize on this implicit uniform random permutation, we introduce an algorithm for \textit{deconstructing uniform random permutations}.
The procedure is an emulation of a generative process in reverse, transforming the uniform random permutation into independent random bits.
In a way, we ``recycle'' the randomness usually discarded during the insertion process.
It is worth noting that the broader paradigm of entropy recycling in random generation is concurrently being explored in some different frameworks, see~\cite{draper2026efficient} for online sampling and \cite[Alg. 5]{DevroyeG15} for recycling bits while sampling batches of i.i.d variables.

The random bits obtained at the end of the generation can then be used to do something else, having an effective entropic random bit consumption (the number of random bit used in the process minus the number of random bits returned) that tightly matches the lower bound.
This allows entropic generation of batches of binary words of size $n$ with $k$ ones, by using the random bits recycled for the previous words to generate the next one.

Alternatively, we can stop the insertions mid process, the smaller uniform random permutation can then be recycled mid-process, and the random bits reinjected into the generation procedure.
Doing this multiple times allow a $1+\varepsilon$ entropic generation for any $\varepsilon > 0$, meaning that the generation algorithm is near-entropic.
The algorithm remains linear.

The remainder of this paper is organized as follows. \Cref{sec:existingAlgos} presents two algorithms for sampling binary words that are we adapt in the subsequent sections. \Cref{sec:recycling} details our ``entropy recycling'' mechanism for uniform random permutations. \Cref{sec:entropicSampling} presents the full sampling algorithm and its complexity analysis.

\section{Existing Algorithm}\label{sec:existingAlgos}

\subsection{The Fisher-Yates Insertion Algorithm}

A standard algorithm for generating a uniform random permutation is the Fisher-Yates shuffle, also known as the Knuth shuffle~\cite{FisherYates}. 
While typically used to permute $n$ distinct elements, the procedure can be adapted to generate random binary words of length $n$ with fixed Hamming weight $k$.

\begin{wrapfigure}{R}{0.4\textwidth}
  \begin{algorithm}[H]
    \SetKwInOut{Input}{Input}
    \SetKwInOut{Output}{Output}
    \DontPrintSemicolon
    \caption{FY-Binary$(n, k)$}
    \label{alg:fyBinary}
    \small
    \Input{$n, k \in \mathbb{N}, k \le n$}
    \Output{Binary word $c \in \{0,1\}^n$}
    \BlankLine
    Initialize $c \leftarrow 0^{n-k}1^k$\;
    \For{$i \leftarrow n$ \textbf{downto} $n-k+1$}{
        $j \leftarrow \text{Uniform}(1, i)$\;
        Swap$(c_i, c_j)$\;
    }
    \Return{$c$}\;
  \end{algorithm}
\end{wrapfigure}

The adaptation exploits the indistinguishability of the 0s. Instead of performing a full shuffle of $n$ elements, we perform a \textit{partial shuffle} restricted to the $k$ ones. The array is initialized with $n-k$ zeros followed by $k$ ones. The algorithm then iterates backwards from $n$ down to $n-k+1$. In each step $i$, the element at position $i$ (guaranteed to be a 1 prior to the swap) is exchanged with an element at a uniformly selected index $j \in \{1, \cdots, i\}$.

% This reduction limits the number of necessary random variates to exactly $k$, drawn from ranges of size $n, n-1, \cdots, n-k+1$. The procedure is summarized in Algorithm~\ref{alg:fy_binary}. While efficient for small $k$, the cost of generating large random variates and the array manipulation becomes significant as $k$ increases.

\subsection{The Balanced Merge Procedure}

The core of the \textsc{MergeShuffle} algorithm from \cite{bacher2015mergeshuffle} is the \textsc{BalanceMerge} procedure, which combines two previously shuffled sub-arrays of size $n/2$ into a single randomly shuffled array of size $n$. 
When the input sub-arrays are of equal size (the ``balanced'' case), this procedure is highly efficient in terms of random bits.
This corresponds to generating a balanced binary word (i.e. with $k=n/2$).

The algorithm proceeds in two distinct phases:
\begin{enumerate}
    \item \textbf{Bit-Driven Selection:} We iterate through the target array. At each step, a random bit is generated (a coin flip). If the bit is 0, we select the next element from the left sub-array; if 1, we select from the right. This loop continues until one of the two sub-arrays is completely exhausted.
    \item \textbf{Fisher-Yates Completion:} Once one sub-array is empty, the remaining elements from the other sub-array are not simply appended. 
    To ensure a uniform distribution, they are inserted into the final array using a partial Fisher-Yates shuffle. Specifically, each remaining element is swapped with a randomly selected position in the valid range processed so far.
\end{enumerate}

When $k=n/2$ generating a uniform binary balanced word using this method consumes $n + \Theta(\sqrt{n} \log n)$ random bits on average. \cite{bacher2015mergeshuffle}.

\begin{wrapfigure}{R}{0.5\textwidth}
  \begin{algorithm}[H]
\SetKwInOut{Input}{Input}
\SetKwInOut{Output}{Output}
\DontPrintSemicolon
\caption{BalancedMerge$(n)$, $n$ even}
\label{alg:balancedMerge}
\tcp{Generates a balanced word}
Let $W \leftarrow \epsilon$, the empty word\;
Let $i_0 \leftarrow 0, \quad i_1 \leftarrow 0$ \tcp*{The numbers of $1$s and $0$s in $W$}
\BlankLine
\tcp{Phase 1: Coin flipping}
\textbf{Loop:}\;
    \eIf{\textnormal{RandomBit}() $= 0$}{
        \eIf{$i_0 = n/2$}{\textbf{break}, let $b \leftarrow 0$}{
        $i_0 \leftarrow i_0 + 1$\;
        $W \leftarrow W\cdot 0$}
    }{
        \eIf{$i_1 = n/2$}{\textbf{break}, let $b \leftarrow 1$}{
        $i_1 \leftarrow i_1 + 1$\;
        $W \leftarrow W\cdot 1$}
    }
\BlankLine
\tcp{Phase 2: Fisher-Yates insertion for remaining elements}
\While{$i_0+i_1 < n$}{
    $W \leftarrow W\cdot b$\;
    $j \leftarrow \text{Uniform}(1, i_0 + i_1)$\;
    Swap$(W[j], W[i_0+i_1])$\;
    $i_b \leftarrow i_b + 1$\;
}
\end{algorithm}
\end{wrapfigure}

\begin{lemma}
\label{lem:fyInsertions}
In the balanced merge procedure called on an even $n$, the number of elements $M$ that remain in the non-depleted sub-array (which are inserted using the Fisher-Yates procedure) follows the Banach Matchbox~\cite[VI. 8. (a)]{Feller1968} distribution, defined for $m \in [0,n/2]$ by:
\[ \mathbb{P}(M = m) = \binom{n - m}{n/2} \frac{1}{2^{n - m}+1}. \]
\end{lemma}

\begin{proof}
It is exactly the same problem than in~\cite[VI. 8. (a)]{Feller1968}.
\end{proof}

\begin{remark}
A consequence is that the expected number of elements that require Fisher-Yates insertion is $\Theta(\sqrt{n})$. 
Because these remaining elements require drawing discrete uniform integers of increasing ranges, the bit cost for this completion phase is $\Theta(\sqrt{n} \log n)$. 
This explains why the \textsc{BalancedMerge} algorithm is entropically effective: it processes the vast majority of elements using one bit per element, and restricts the more expensive Fisher-Yates fallback to a vanishingly small $O(\sqrt{n})$ fraction of the input.
\end{remark}

\begin{remark}
The \textsc{BalancedMerge} algorithm can be modified to sample binary words that are not balanced by using Bernoulli variables of parameter $k/n$ instead of coin flips.
This leads to linear time, near-entropic algorithms for binary words, as long as $k$ is not too small (as long as $k = \Omega(n^{2/3 +\varepsilon})$, with some $\varepsilon > 0$ to be precise) and all the Bernoulli variables are generated in a near-entropic way.
This last part can be achieved by generating the Bernoulli variables in fixed sized batch, say by using \cite{saad2020fast}.
\end{remark}

\section{Random Bit Recycling}\label{sec:recycling}

\subsection{Random Bit Extraction}

In the standard paradigm of random generation, an algorithm consumes a stream of independent, uniform random bits to produce a combinatorial structure~\cite{Knuth1976TheCO}.
Here, we invert this process: in the case of a uniform random permutation, we map the uniform combinatorial structure back into a sequence of independent, uniform random bits. 
We call this \emph{random bit extraction} or \emph{random bit recycling} as we give back some of the bits used to generate the structure.

Consider a uniform balanced binary word $W \in \{0,1\}^{n}$, which consists of exactly $n/2$ zeros and $n/2$ ones. 
The entropy of this structure is $\log_2 \binom{n}{n/2}$, which by Stirling's approximation is $n - \frac{1}{2}\log_2(\pi n) + O(1)$. 
Our goal is to extract a sequence of independent random bits from $W$ whose expected length is as close as possible to $\mathcal{H}$, the entropy of the structure, which, by Shannon's source coding theorem~\cite{shannon1948mathematical}, is the maximum amount that one could extract.

We achieve this by conceptually reversing the generative \textsc{BalancedMerge} procedure.
To do so, we first sample the simulated depletion condition: the number of remaining elements $M$ whose distribution is given by \cref{lem:fyInsertions}, then emulate in reverse the $M$ Fisher-Yates insertions.
The resulting word is composed of independent unbiased bits.

The full extraction procedure is detailed in Algorithm~\ref{alg:extractBits}.

\begin{lemma}\label{lem:recycleMerge}
    When called on $W$ a uniform balanced binary word of length $n$, \cref{alg:extractBits} runs in linear time, consumes $O(\sqrt{n}\log n)$ random bits and returns $n - O(\sqrt{n})$ uniform\footnote{Standard $O(\cdot)$ notation absorbs the sign, making $+ O(f(n))$ and $- O(f(n))$ equivalent.
    However, we retain the sign to indicate the direction of the deviation.}, independent bits.
\end{lemma}

\begin{proof}[Proof of Lemma \ref{lem:recycleMerge}]
    Let $W$ be distributed uniformly over all $\binom{n}{n/2}$ balanced words.
    We denote by $\Tilde{W}$ the returned binary prefix of length $n-M$. 
    We condition our analysis of the loop invariant on the event $(b=0)$ (the case $(b=1)$ is symmetric). 
    Let $W^{(i)}$ be the prefix of length $L = n-i$ after $i$ iterations of the extraction loop. 

    \textbf{Claim:} For any $i \in \{0, \dots, M\}$, conditional on $b=0$, the prefix $W^{(i)}$ is uniformly distributed among all binary words of length $L$ containing exactly $n/2$ ones.
    
    \textit{Proof of claim:} We proceed by induction. 
    The base case $W^{(0)} = W$ is trivial. 
    Assume $W^{(i)}$ is uniformly distributed among all binary words of length $L$ with exactly $n/2$ ones. 
    The algorithm selects an index $j$ uniformly from the $n/2 - i$ positions of '0', swaps $W^{(i)}_j$ with $W^{(i)}_L$, and truncates the word to length $L-1$.
    Let $y$ be a specific target word of length $L-1$. 
    To obtain $y$, the state before truncation must be $y \cdot 0$.
    There are exactly $L = n-i$ valid ways to choose an initial word $x$ and a distinguished '0' that yield $y \cdot 0$ post-swap: $n/2 - i$ ways by choosing $x = y\cdot 0$, and $n/2$ ways by choosing an $x$ ending in '1'. 
    Because the prior word and the index choice are uniform, each configuration has probability $((n/2-i)\binom{L}{n/2})^{-1}$. 
    Therefore, the probability of obtaining $y$ is exactly $L \times ((n/2-i)\binom{L}{n/2})^{-1} = \frac{L}{n/2-i} \frac{(n/2)!(L-n/2)!}{L!} = \binom{L-1}{n/2}^{-1}$. 
    This confirms the conditional distribution of $W^{(i+1)}$ remains uniform, proving the claim.
    
    \textbf{Net bit extraction:} The variable $M$ (following the Banach Matchbox distribution) and the uniform coin flip $b$ are sampled independently, yielding the joint probability $\mathbb{P}(M=m \cap b=0) = \binom{n-m}{n/2} 2^{-(n-m)} \times \frac{1}{2} = \binom{n-m}{n/2} 2^{-(n-m+1)}$. 
    Let $w$ be a specific binary sequence of length $n-m$ containing exactly $n/2$ ones. 
    By our previous claim, the conditional probability of obtaining $w$ is exactly $\mathbb{P}(\tilde{W} = w \mid M = m \cap b=0) = \binom{n-m}{n/2}^{-1}$. 
    Multiplying these gives the joint probability of extracting the exact prefix $w$ alongside $b=0$, which cancels the binomial coefficient: $\mathbb{P}(\tilde{W} = w \cap b=0) = \binom{n-m}{n/2}^{-1} \times \binom{n-m}{n/2} 2^{-(n-m+1)} = 2^{-(n-m+1)}$.
    This proves that the extracted prefix $\tilde{W}$, when concatenated with the independent bit $b$, yields a sequence of exactly $n-M+1$ perfectly independent and uniform random bits.
    
    The algorithm runs in $O(n)$ time.
    As the expected value of the Banach Matchbox distribution yields $\mathbb{E}[M] = O(\sqrt{n})$, this procedure successfully extracts $n - O(\sqrt{n})$ bits on average.
    Sampling $M$ takes on average $O(\log n)$ bits to run using the method in \cite{devroye1987simple}, and each sample of $\text{Uniform}(S_b)$ takes $O(\log n)$ bits and time, so the procedure consumes $O(\sqrt{n}\log n)$ random bits, for a net gain of $n - O(\sqrt{n}\log n)$.
\end{proof}

\begin{algorithm}[htbp]
\SetKwInOut{Input}{Input}
\SetKwInOut{Output}{Output}
\DontPrintSemicolon
\caption{ExtractBits$(W)$}
\label{alg:extractBits}
\small
\Input{A uniform balanced binary word $W \in \{0,1\}^{n}$ (exactly $n/2$ zeros and $n$ ones)}
\Output{A sequence of independent uniform random bits}
\BlankLine
\tcp{1. Sample the simulated depletion state}
Sample $M \in \{0, \cdots, n/2\}$ according to the Matchbox distribution\footnotemark: $\mathbb{P}(M=m) = \binom{n-m}{n/2} 2^{-(n-m)}$\;
Sample a uniform bit $b \leftarrow \text{Uniform}(\{0, 1\})$\;
\BlankLine
\tcp{2. Reverse Fisher-Yates Extraction}
\For{$i \leftarrow 0$ \KwTo $M - 1$}{
    \tcp{Find all active occurrences of bit $b$}
    Let $S_b \leftarrow \{k \in \{1, \cdots, n-i\} \mid W_k = b\}$\;
    
    \tcp{Select one occurrence uniformly at random}
    $j \leftarrow \text{Uniform}(S_b)$\; 
    
    \tcp{Swap it to the current end of the active word}
    Swap$(W_j, W_{n-i})$\label{line:swap}\;
}
\BlankLine
\tcp{3. The remaining prefix contains the unconditioned bits}
\Return{$W[1 \cdots n-M] \cdot b$}\;
\end{algorithm}

\footnotetext{One can easily show that the matchbox distribution is log-concave, hence it is possible to sample $M$ very efficiently, using the algorithm described in \cite{devroye1987simple}.}

\subsection{Recursive Deconstruction of Uniform Permutations}

Using the recycling mechanism for balanced binary words, we deconstruct uniform random permutations.
The core idea is to split the permutation of size $n$ into one balanced binary word and two uniform permutations of size $n/2$.
This gives algorithm~\ref{alg:deconstructPerm}.
\begin{wrapfigure}{R}{0.5\textwidth}
\begin{algorithm}[H]
\SetKwInOut{Input}{Input}
\SetKwInOut{Output}{Output}
\DontPrintSemicolon
\caption{DeconstructPermutation$(\pi)$}
\label{alg:deconstructPerm}
\small
\Input{A uniform random permutation $\pi$ of size $n \geq 2$}
\Output{A stream of independent uniform random bits}
\BlankLine
\If{$n = 2$}{
    \textbf{Yield} $0$ if $\pi = (1, 2)$, $1$ if $\pi = (2, 1)$ \tcp*{Base case}
}
\BlankLine
\tcp{1. Odd Size Reduction}
\If{$n$ is odd}{
    $\pi' \leftarrow \pi$ with the maximum element removed\;
    \textbf{Yield} {\textnormal{DeconstructPermutation}$(\pi')$}\;
}
\BlankLine
\tcp{2. Even Size Decomposition ($n$)}
Initialize an empty word $W \in \varepsilon$\;
Initialize empty sequences $\pi_{top}, \pi_{bot}$\;
\BlankLine
\For{$i \leftarrow 1$ \KwTo $n$}{
    \eIf{$\pi[i] > n/2$}{
        $W[i] \leftarrow 1$\;
        $\pi_{top} \leftarrow \pi_{top} \cdot (\pi[i] - n/2)$\;
    }{
        $W[i] \leftarrow 0$\;
        $\pi_{bot} \leftarrow \pi_{bot} \cdot (\pi[i])$\;
    }
}
\BlankLine
\tcp{3. Bit Extraction and Recursion}
$Bits \leftarrow \textnormal{ExtractBits}(W)$\;
\textbf{Yield} $Bits$\;
\BlankLine
DeconstructPermutation$(\pi_{top})$\;
DeconstructPermutation$(\pi_{bot})$\;
\end{algorithm}
\end{wrapfigure}

The entropy of such a uniform permutation is $\log_2(n!) = n \log_2 n - O(n)$. 
We demonstrate that our algorithm achieves this optimal bound up to a linear error term.

\begin{theorem}
\label{thm:totalRecycledBits}
Algorithm~\ref{alg:deconstructPerm} has a running time in $O(n\log n)$ and if we denote $B(n)$ the expected number of independent uniform random bits extracted from a uniform random permutation of size $n$ :
\begin{equation*}
    B(n) = n\log_2 n - O(n).
\end{equation*}
\end{theorem}

\begin{proof}
By the master theorem for divide-and-conquer algorithms (see \cite[Thm. 2.3]{roura2001improved} for a formal statement), as the toll function is linear and the two sub-problems are of size $n/2$ the runtime of algorithm~\ref{alg:deconstructPerm} is in $O(n\log(n))$.
For the random bit extraction, the result is an application of the more complex \cite[Thm. 1]{drmota2013master}, with $B(n) = n + 2B(\lfloor n/2 \rfloor) - O(\sqrt{n}\log n)$.
\end{proof}

% \begin{remark}
% In Algorithm~\ref{alg:deconstruct_perm}, when the permutation size $N = 2n+1$ is odd, the position of the maximum element is simply yielded as a bounded uniform integer in $\{1, \cdots, 2n+1\}$. It is possible--albeit non-trivial--to recycle this uniform integer directly into a (small) stream of independent random bits. \red{supprimer cette remarque ?}
% \end{remark}

\section{Entropic Sampling of Sparse Binary Words}\label{sec:entropicSampling}

\subsection{Extracting Fisher-Yates Entropy via Distinguishable Insertions}

Morally, the reason the partial Fisher-Yates shuffle fails to be entropic for generating combinations is that it implicitly generates an ordering on the $1$s that are inserted, from their order of insertion.
This wastes $\log_2(k!)$ random bits on an internal ordering that is ultimately ignored. 
We can recover this ``wasted'' entropy by making this ordering explicit, extracting it, and recycling it. 

Instead of initializing the array with indistinguishable $1$s, we initialize it with distinct identifiers $1, 2, \cdots, k$. 
We then run the partial Fisher-Yates shuffle as before. 
Once the $k$ insertions are complete, the array indices containing non-zero elements dictate the generated binary word. 
Additionally, the sequence of these non-zero elements, when read from left to right, forms a uniform random permutation of the set $\{1, \cdots, k\}$. 

By isolating this permutation and feeding it into the \textnormal{DeconstructPermutation} algorithm, we ``give back'' the excess random bits consumed during the generation phase. 
This coupling of Fisher-Yates generation with recursive deconstruction yields an effective algorithm for the polynomially sparse regime.

\begin{theorem}
\label{thm:optimalFyRecycling}
We assume $k \sim n^\beta$ with a fixed $\beta \in (0,1)$. 
The generation of a binary word of length $n$ and Hamming weight $k$ via distinguishable Fisher-Yates insertion followed by \textnormal{DeconstructPermutation} has an expected runtime of $O(n)$. 
Furthermore, the algorithm is first order entropic: the expected net bit consumption (bits consumed minus bits recycled) is $\log_2 \binom{n}{k} + O(k)$.
\end{theorem}

\begin{proof}
\textbf{Time Complexity:} 
Performing the $k$ Fisher-Yates insertions takes $O(k \log n)$ time, while extracting the random bits from the resulting permutation of size $k$ requires $O(k \log k)$ time.
Because $k \sim n^\beta$ for $\beta \in (0,1)$, $k\log n \sim n^\beta \log n = O(n)$ and $k\log k = O(n)$, hence both steps are bounded by $O(n)$, so the overall runtime is in $O(n)$. 
One can note that this cost could be decreased to sub-linear time if the algorithm returned the sparse set of indices for the $1$s rather than explicitly writing out the complete binary word.

% \textbf{Net Bit Complexity:}
% The partial Fisher-Yates shuffle requires $k$ independent draws from discrete uniform distributions of decreasing sizes: $n, n-1, \cdots, n-k+1$.
% Using the algorithms described in \cite{lumbroso2013optimal,saad2020fast}, sampling a uniform integer in $[i]$ takes $\log(i) + O(1)$ random bits.
% The expected number of random bits consumed is therefore 
% \begin{equation*}
%     \sum_{i=n-k+1}^n \log_2 i  + O(1) = \log_2\left(\frac{n!}{(n-k)!}\right) + O(k).
% \end{equation*}

% The algorithm then recycles a uniform permutation of size $k$, by theorem~\ref{thm:total_recycled_bits}, the expected number of random bits returned to the randomness pool is:
% \begin{equation*}
%     k \log_2 k - O(k) = \log_2(k!) - O(k).
% \end{equation*}

% So all in all, the expected net bit consumption is $\log_2 \left(\frac{n!}{(n-k)!}\right) - \log_2(k!) + O(k) = \log_2 \binom{n}{k} + O(k)$ which matches the entropy of the structure up to $O(k)$ (and $O(k) = o(\log_2 \binom{n}{k})$ for our range of $k$).

\textbf{Net Bit Complexity:} The partial Fisher-Yates shuffle performs $k$ independent draws from discrete uniform distributions of decreasing sizes $n, \dots, n-k+1$. 
Using optimal generation algorithms~\cite{lumbroso2013optimal,saad2020fast}, sampling in $[i]$ requires $\log_2 i + O(1)$ bits. 
The expected total consumption is therefore $\sum_{i=n-k+1}^n \log_2 i + O(k) = \log_2(\frac{n!}{(n-k)!}) + O(k)$. 
The algorithm subsequently recycles a uniform permutation of size $k$, returning $k \log_2 k - O(k) = \log_2(k!) - O(k)$ random bits to the pool (Theorem~\ref{thm:totalRecycledBits}). 
Overall, the net bit consumption is exactly $\log_2(\frac{n!}{(n-k)!}) - \log_2(k!) + O(k) = \log_2 \binom{n}{k} + O(k)$, which asymptotically matches the entropy of the structure since $O(k) = o(\log_2 \binom{n}{k})$ in our target regime.

\end{proof}

\subsection{Near Entropic Binary Word Sampling}

We now present a linear-time, near-entropic algorithm that generates a uniform binary word of length $n$ with $k \sim n^\beta$ $1$s (where $\beta \in (0,1)$). 
Algorithm~\ref{alg:chainedFy} simply consumes a small number of random bits, and it does not ``give back'' random bits.
It works by chaining generation and deconstruction phases across $l+1$ discrete steps.
The entropy ``wasted'' on the internal ordering of the $1$s inserted during step $i$ is extracted and recycled to power the insertions of step $i+1$. 
We define the constant $\alpha := \frac{1-\beta}{1-\beta^{l+1}}$. 
On the $i$-th step (for $i=0, \cdots, l$), we insert $k_i = \alpha \beta^i k$ elements. 
The procedure is described in algorithm~\ref{alg:chainedFy}.

\begin{algorithm}[htbp]
\SetKwInOut{Input}{Input}
\SetKwInOut{Output}{Output}
\DontPrintSemicolon
\caption{$l$-ChainedRecyclingFY$(n, k)$}
\label{alg:chainedFy}
\small
\Input{Word length $n$, $k$ such that $k \sim n^\beta$}
\Output{A uniform binary word $W \in \{0,1\}^n$ of Hamming weight $k$}
\BlankLine
$\alpha \leftarrow \frac{1-\beta}{1-\beta^{l+1}}$\;
Initialize $W \leftarrow 0^{n-k}$\;

\BlankLine
\For{$i \leftarrow 0$ \KwTo $l$}{
    $k_i \leftarrow \alpha \beta^i k$
    
    \eIf{$i = 0$}{
        $Bits \leftarrow \varepsilon$ the empty word\;
    }{
        \tcp{Recycle the permutation from the previous step}
        $Bits \leftarrow Bits\cdot \textnormal{DeconstructPermutation}(\pi_{i-1})$\;
    }
    
    \tcp{Insert distinguishable elements into the remaining zeros}
    Insert elements $\{1, \dots, k_i\}$ into $W$ via Fisher-Yates, consuming the random bits from $Bits$ and drawing fresh bits only when $Bits$ is exhausted\;
    
    \tcp{Extract their spatial ordering as a uniform permutation}
    $\pi_i \leftarrow$ The sequence of the newly inserted $k_i$ elements read from left to right\;
    
    \tcp{Standardize the inserted elements}
    Replace all elements $\{1, \cdots, k_i\}$ in $W$ with $1$\;
}
\Return{$W$}\;
\end{algorithm}

\begin{theorem}
\label{thm:chainedRecycling}
We assume $k \sim n^\beta$ with a fixed $\beta \in (0,1)$.
For any constant integer $l \ge 0$, \cref{alg:chainedFy} generates a uniform binary word of length $n$ and Hamming weight $k$ in $O(n)$ time.
The total cost in random bits is
$$ \frac{1-\beta}{1-\beta^{l+1}} k \log_2 n + O(k), \quad \text{while} \quad \log_2\binom{n}{k} = (1-\beta) k \log_2 n + O(k)$$
\end{theorem}

\begin{proof}
\textbf{Time Complexity:} 
The algorithm performs $l+1 = O(1)$ steps. Step $i$, runs in $O(k\log n)$ + $O(k \log k)$ as it inserts $O(k)$ elements and then recycles a permutation of size $O(k)$. 
Since $k = n^\beta$ for a constant $\beta < 1$, we have $k \log n = O(n)$, hence the overall expected runtime in $O(n)$.

\textbf{Bit Complexity:} 
To insert $k_i$ elements into an array of size $\le n$, the Fisher-Yates procedure requires $k_i \log_2 n + O(k_i)$ bits. 
Once inserted, these $k_i$ elements form a uniform random permutation. 
Deconstructing this permutation yields $\log_2(k_i!) = k_i \log_2 k_i - O(k_i)$ independent random bits.

% Since $k = n^\beta$, we have $k_i = \alpha \beta^i n^\beta$. Thus, the number of bits recycled from step $i$ is:
% \begin{align*}
% k_i \log_2(k_i) &= k_i \log_2(\alpha \beta^i n^\beta) \\
% &= k_i \left( \beta \log_2 n + \log_2(\alpha \beta^i) \right) \\
% &= \beta k_i \log_2 n - O(k_i)
% \end{align*}
% Crucially, because $k_{i+1} = \beta k_i$, the recycled bits amount to $k_{i+1} \log_2 n - O(k_i)$. 

% Therefore, summing the fresh bits across all stages gives:
% \begin{equation*}
% \text{Total Fresh Bits} 
% = k_0 \log_2 n + O(k_0) + \sum_{i=1}^l O(k_i) 
% = \alpha k \log_2 n + O(k).
% \end{equation*}
% With $\alpha = \frac{1-\beta}{1-\beta^{l+1}}$ it yields the stated bit complexity.

Since $k = n^\beta$, we have $k_i = \alpha \beta^i n^\beta$. 
Thus, the number of bits recycled from step $i$ is $k_i \log_2(k_i) = \beta k_i \log_2 n - O(k_i)$. 
Crucially, because $k_{i+1} = \beta k_i$, the recycled bits amount to exactly $k_{i+1} \log_2 n - O(k_i)$. Therefore, summing across all stages gives a total of $k_0 \log_2 n + O(k_0) + \sum_{i=1}^l O(k_i) = \alpha k \log_2 n + O(k)$. With $\alpha = \frac{1-\beta}{1-\beta^{l+1}}$, this yields the stated bit complexity.
\end{proof}

\begin{remark}
Theorem~\ref{thm:chainedRecycling} implies that for any desired tolerance $\varepsilon > 0$, we can bound the multiplicative over-consumption of random bits to $1+\varepsilon$ by demanding: $\frac{1}{1-\beta^{l+1}} \le 1+\varepsilon.$
I.e. by taking $l = \left\lceil \frac{-\log(1+1/\varepsilon)}{\log \beta} - 1 \right\rceil$, $l$-\textsc{ChainedRecyclingFY} is an $\varepsilon$-entropic, linear time generation algorithm.
\end{remark}

\section{Conclusion}\label{sec:conc}

We presented a linear-time (in the length of the generated word) algorithm for the uniform generation of fixed Hamming weight binary words in the previously untreated polynomially sparse regime.
Our method achieves a random bit consumption that matches the information-theoretic lower bound up to a multiplicative factor of $1+\varepsilon$ for a chosen $\varepsilon > 0$.

The algorithm is based on a random bit recycling procedure. 
We believe this paradigm of entropy deconstruction is generalizable to other combinatorial contexts.
Most notably, we are currently working to adapt this mechanism to recycle the ``wasted'' entropy from the rejected structures that inherently occur during exact-size generation with Boltzmann-type samplers. 
We believe that there is practical relevance to this recycling paradigm: it allows for entropic generation when sampling structures in large batches (by pooling recovered bits across iterations), or almost-entropic generation for single instances when the recycling is chained recursively within the generation procedure itself like done in algorithm~\ref{alg:chainedFy}.

More broadly, we believe that evaluating randomized algorithms through the lens of random bit consumption provides a more precise analytical framework than standard runtime analysis alone.
Treating random bits as a resource, can expose inefficiencies of standard methods and provides a guide for the conception of new algorithms.

\textbf{Acknowledgments: }
The authors would like to express their gratitude to Étienne Objois for many valuable and careful comments. 
This research was supported by anr-fwf project PAnDAG ANR-23-CE48-0014-01.

% \nocite{*}
\bibliographystyle{eptcs}
\bibliography{generic}
\end{document}